\documentclass{llncs}

\usepackage{amssymb,stmaryrd,mathrsfs,dsfont,mathtools}

\input xy
\xyoption{all}

\usepackage{tikz}
\tikzstyle{vertex}=[circle, draw, inner sep=0pt, minimum size=6pt]
\newcommand{\vertex}{\node[vertex]}

\begin{document}
\frontmatter          
\pagestyle{headings}  

\mainmatter              

\title{On the Representation Number of Bipartite Graphs}

\titlerunning{On the Representation Number of Bipartite Graphs}  
%

\author{Khyodeno Mozhui \and 
K. V. Krishna} 

\authorrunning{Khyodeno Mozhui \and K. V. Krishna} 

\institute{Institute of Technology Guwahati, Guwahati, India,\\
	\email{k.mozhui@iitg.ac.in};\;\;\; 
	\email{kvk@iitg.ac.in}}

\maketitle              

\begin{abstract}
A word-representable graph is a simple graph $G$ which can be represented by a word $w$ over the vertices of $G$ such that any two vertices are adjacent in $G$ if and only if they alternate in $w$. It is known that the class of comparability graphs -- the graphs which admit a transitive orientation -- is precisely the class of graphs that can be represented by a concatenation of permutations of vertices. The class of bipartite graphs is a subclass of comparability graphs.   While it is an open problem to determine the representation number of comparability graphs, it was conjectured that the representation number of bipartite graphs on $n$ vertices is at most $n/4$. In this paper, we propose a polynomial time relabeling algorithm to produce a word representing a given bipartite graph which is a concatenation of permutations of the graph's vertices. Thus we obtain an upper bound for the representation number of bipartite graphs, which in turn gives us an upper bound for the dimension of the posets corresponding to bipartite graphs.
\end{abstract}


\section{Introduction}

A word-representable graph is a simple graph $G$ which can be represented by a word $w$ over the vertices of $G$ such that any two vertices are adjacent in $G$ if and only if they alternate in $w$. While the origin of word-representable graphs lies in the study of Perkins semigroups \cite{kitaev08order}, this notion covers many important classes of graphs including comparability graphs, circle graphs and 3-colorable graphs. The literature has several important contributions to the theory of word-representable graphs and their connections to various concepts. A comprehensive introduction to word-representable graphs and their connections and contributions can be found in the monograph \cite{kitaev15mono}.

The representation number of a word-representable graph is the minimum number $k$ such that the graph is represented by a word containing exactly $k$ copies of each letter. The class of complete graphs is the class of graphs with representation number 1 and it was established in \cite{halldorsson11} that the class of circle graphs characterizes the graphs with representation number 2. Although there are examples of graphs with higher representation number, no characterization of graphs is available for other numbers. 

The class of comparability graphs -- the graphs which admit a transitive orientation -- play a vital role in the theory of word-representable graphs. In fact, the neighborhood of each vertex in a word-representable graph is a comparability graph \cite{kitaev08}. Furthermore, the class of comparability graphs is precisely the class of permutationally representable graphs, that is, they can be represented by a concatenation of permutations of their vertices \cite{kitaev08order}. We call the minimum number of such permutations of vertices of a comparability graph its permutation representation number. Every comparability graph corresponds to a partially ordered set (poset) based on one of its transitive orientations. Due to their correspondence with posets, comparability graphs received importance in the literature. It is known that finding the permutation representation number of a comparability graph is equivalent to finding the dimension of the corresponding poset \cite{halldorsson11}. However, the problem of finding the dimension of a poset is NP-hard \cite{yanna82}.

The class of bipartite graphs is a subclass of the class of comparability graphs.   
While it is an open problem to determine the representation number of a comparability graph, it was conjectured in \cite{glen18} that the representation number of a bipartite graph on $n$ vertices is at most $n/4$. In this paper, we propose a polynomial time relabeling algorithm to produce a word that is a concatenation of permutations of vertices representing a given bipartite graph. Thus we obtain an upper bound for the permutation representation number of bipartite graphs and consequently an upper bound for their representation number. This indeed gives us an upper bound for the dimension of the posets corresponding to bipartite graphs. In fact, the upper bound for the permutation representation number is tight for the class of bipartite graphs.

\section{Preliminaries}

In this section, we provide necessary background material on words, graphs, and word-representable graphs. For the concepts that are not presented here, we may refer to  \cite{kitaev15mono,west}.   

Let $A$ be a finite set. A word over $A$ is a finite sequence of elements of $A$. The words are written by juxtaposing the symbols of the sequence. The empty sequence, called the empty word, is denoted by $\varepsilon$. A subword $u$ of a word $w$ is a subsequence of the sequence $w$ and it is denoted by $u \le w$. For instance, $abaa \le aabbaba$. Clearly, $\le$ is a partial order relation on the set of all words over $A$. We say the letters $a$ and $b$ alternate in the word $w$ if the subword consisting of all occurrences of $a$ and $b$, that is the subword obtained by deleting all other letters of $w$, is of the form $abab\cdots$ or $baba\cdots$ (of even or odd length). A word $w$ is said to be $k$-uniform if the number of occurrences of each letter in $w$ is $k$. 

Let $A = \{a_{i_1}, a_{i_2}, \ldots, a_{i_k}\}$ be a set and $i_1 < i_2 < \cdots < i_k$. We write $\text{Dec}(A)$ to represent the word $a_{i_k} a_{i_{k - 1}} \cdots a_{i_1}$ in which the symbols of $A$ appear exactly once and are arranged such that the subscripts are in decreasing order. 

In this paper, the term graph represents only a simple and undirected graph, i.e., a graph without loops or parallel edges. To distinguish between two-letter words and edges of graphs, we write $\overline{ab}$ to denote an undirected edge between vertices $a$ and $b$. A directed edge from $a$ to $b$ is denoted by $\overrightarrow{ab}$. The neighborhood of a vertex $a$ in a graph, denoted by $N(a)$, is the set of all vertices adjacent to $a$ in the graph. An orientation of a graph is an assignment of direction to each edge so that the resulting graph is a directed graph. An orientation of a graph is said to be a transitive orientation, if the adjacency relation on the vertices in the resulting directed graph is transitive. A graph $G$ is said to be a bipartite graph if its vertex set can be partitioned into $\{V_1, V_2\}$, called bipartition, such that every edge in $G$ connects a vertex in $V_1$ to a vertex in $V_2$. A bipartite graph with bipartition $\{V_1, V_2\}$, where $|V_1| = m$ and $|V_2| = n$, is said to be complete, denoted by $K_{m,n}$, if every vertex in $V_1$ is adjacent to all vertices of $V_2$. 

A graph $G = (V, E)$ is said to be word-representable if there is a word $w$ over the set $V$ such that two vertices $a$ and $b$ are adjacent in $G$ if and only if $a$ and $b$ alternate in $w$. A graph is said to be $k$-word-representable if there is $k$-uniform word representing the graph. Every word-representable graph is $k$-word-representable for some $k$ \cite{kitaev08}. The representation number of a word-representable graph $G$, denoted by $\mathcal{R}(G)$, is the smallest number $k$ such that $G$ is $k$-word-representable. A graph $G$ is said to be permutationally representable if it can be represented by a word of the form $w_1w_2\cdots w_k$, where each $w_i$ is a permutation of vertices of $G$. Furthermore, if a graph is permutationally representable involving $k$ permutations, then the graph is called permutationally $k$-representable. 

\begin{definition}
	Let $G$ be a permutationally representable graph. The permutation representation number of $G$, denoted by $\mathcal{R}^p(G)$, is the smallest $k$ such that $G$ is permutationally $k$-representable. 
\end{definition}

\begin{remark}
	If $G$ is a permutationally representable graph, then $\mathcal{R}(G) \le \mathcal{R}^p(G)$.
\end{remark}

The class of complete graphs has permutation representation number 1, i.e., $\mathcal{R}^p(K_n) = \mathcal{R}(K_n) = 1$ for all $n$, where $K_n$ is the complete graph on $n$ vertices.   It was shown in \cite[see the proof of Theorem 4]{halldorsson11} that $\mathcal{R}^p(H_{n,n}) = n$ (for $n \ge 2$) and in \cite{glen18} that $\mathcal{R}(H_{n,n}) = \lceil n/2 \rceil$  (for $n \ge 5$), where $H_{n,n}$ is the graph, called a crown graph, obtained by removing a perfect matching in $K_{n,n}$. Further, it was conjectured in \cite[see Conjecture 1 in page 93]{glen18} that every bipartite graph on $n$ vertices has representation number at most $n/4$.

A graph is said to be a comparability graph if it admits a transitive orientation. Every comparability graph is word-representable. In fact, the class of comparability graphs is precisely the class of permutationally representable graphs \cite{kitaev08order}. Every bipartite graph is a comparability graph. Based on its transitive relation, every comparability graph corresponds to a partially ordered set (poset). In fact, 
the class of bipartite graphs is precisely the class of comparability graphs that are isomorphic to the Hasse diagrams of the corresponding posets. 
 
\begin{remark}
	 Let $G$ be a comparability graph and $H$ be the Hasse diagram of the corresponding poset. If the graphs $G$ and $H$ are isomorphic, then  $G$ is a bipartite graph, and vice versa.   
\end{remark}

The dimension of a poset is the minimum number of linear orders of its elements such that their intersection is the partial order of the poset. It was shown in \cite{halldorsson11} that a comparability graph $G$ is permutationally $k$-representable if and only if the poset corresponding to $G$ has dimension at most $k$. However, since the problem of finding the dimension of a poset is NP-hard \cite{yanna82}, finding the permutation representation number of a comparability graph is NP-hard. Hence, in the direction of finding the permutation representation number of some subclasses of comparability graphs, in \cite{halldorsson11}, the authors obtained the permutation representation number of crown graph $H_{k,k}$. Here, the authors constructed linear orders whose intersection happens to be the corresponding poset of $H_{k,k}$. In the following section, we extend the idea to arbitrary bipartite graphs and construct linear orders of the elements of corresponding poset.

\section{Main Result}

In this section, we devise an algorithmic procedure to construct a word representing permutationally a given bipartite graph. Subsequently, we obtain an upper bound for the permutation representation number of bipartite graphs. Consequently, in view of the correspondence between comparability graphs and posets, we also present an upper bound for the dimension of a special class of posets which correspond to bipartite graphs.

\begin{remark}\label{isolated}
The permutation representation number of a graph does not change by the inclusion or deletion of isolated vertices. 
If a graph $G = (V, E)$ (with at least 3 vertices) has isolated vertices, say $a_1, a_2, \ldots, a_t$, then consider the induced subgraph $H$ of $G$ on the vertex set $V \setminus \{a_1, a_2, \ldots, a_t\}$. That is, $H$ is obtained by deleting the vertices $a_1, a_2, \ldots, a_t$ from $G$. Note that $G$ is permutationally representable if and only if $H$ is permutationally representable. In fact, if $w_1$, $w_2$, \ldots, $w_k$ are permutations of vertices of $H$ such that the word $w_1w_2 \cdots w_k$ represents the graph $H$ permutationally, then we can see that $G$ is also permutationally $k$-representable. For instance, the following word $w$ with $k$ permutations on the vertices of $G$ represents $G$ permutationally: 
\[w  =  \begin{cases*}
			w_1ur(u)w_2w_3ur(u)w_4 \cdots w_ku,\; \text{ if $k$ is odd};\\
			w_1ur(u)w_2w_3ur(u)w_4 \cdots r(u)w_k,\; \text{ if $k$ is even,}	
		\end{cases*}\]
where $u = a_1a_2 \cdots a_t$ and $r(u) = a_ta_{t-1} \cdots a_1$, the reversal of $u$.       
\end{remark}

In what follows, $G$ always denotes a bipartite graph with bipartition $\{V_1, V_2\}$ and $m = |V_1| \le |V_2|$. In view of Remark \ref{isolated}, we assume that $G$ has no isolated vertices.

\subsection{Relabeling algorithm}

We now present an algorithm to determine permutations of vertices of a given bipartite graph $G$ such that their concatenation represents $G$. It is achieved by relabeling the vertices of $G$ in a particular order based on their nonadjacent vertices and construct certain permutations of the vertices of $G$.  The algorithm is presented in the following.

\vspace{.5cm}
\hrule
\vspace{.5cm}
\begin{enumerate}
	\setlength\itemsep{.5em}
	\item Given a bipartite graph $G$, suppose $V_1 = \{a_1, \ldots, a_m\}$ and $V_2 = \{b_1, \ldots, b_n\}$.
	
	\item If $N(a) = V_2$ for all $a \in V_1$, then consider the following word and exit:
		\begin{center}
			$w = a_1 a_2 \cdots \ a_m b_1 b_2 \cdots b_n a_m a_{m-1} \cdots a_1  b_n  b_{n-1} \cdots \ b_1$.
		\end{center}
	
	\item Else, choose $a$ in $V_1$ such that $V_2 \setminus N(a) \ne \varnothing$ and label it as $c_1$.
	
	\item Relabel the remaining vertices of $V_1$ as $c_2, \ldots, c_m$, arbitrarily.
	
	\item Relabel the vertices in $V_2$ as follows:
		\begin{enumerate}
		\item Set $	m' = m$; $A = V_2$.
		
		\item For $i = 1$ to $m$,
			\begin{enumerate}
				\item Let $A_i = A \setminus N(c_i)$, say $\{b_{j_1},...,b_{j_{k_i}}\}$. 			
				
				\item Relabel the vertices of $A_i$ as $c_{m' + 1}, \ldots, c_{m' + k_i}$, arbitrarily.  
				
				\item  Set $m' = m' + k_i$; $A = A \setminus A_i$.
			\end{enumerate}
		\item If $A \ne \varnothing$, then relabel the remaining vertices of $A$ as $c_{m' + 1}, \ldots, c_{m + n}$, arbitrarily.
		\end{enumerate}
	
	\item Create a list of permutations of the vertices $c_1, \ldots, c_{m + n}$ as per the following:
		\begin{enumerate}
			\item Set $w_1 =  c_{m} c_{m-1} \cdots c_{2} c_{m+1} c_{m+2}  \cdots c_{m+k_1} c_{1} c_{m+k_1+1}  c_{m+k_1+2}  \cdots  c_{m+n}$.
			
			\item For $i = 2$ to $m$, 
			\begin{enumerate}
				\item[] If $V_2 \setminus N(c_i) \ne \varnothing$, then set
				
				\qquad $w_i = c_1 \cdots c_{i - 1} c_{i + 1} \cdots c_m \text{Dec}(V_2 \setminus N(c_i)) c_{i} \text{Dec}(N(c_i))$;
				
				else, set $w_i =  \varepsilon$.
			\end{enumerate}
		
			\item If $N(c) = V_2$ for some $c \in V_1$ then set 
			
				\qquad $w_{m + 1} =  c_1 c_2 \cdots c_m c_{m+n} \cdots c_{m+1}$;
				
				else, set $w_{m + 1} =  \varepsilon$.
		\end{enumerate}	
		
	\item  Report the word $w$ obtained by replacing the original labels of the vertices of $G$ in the word $w_1 \cdots w_m w_{m + 1}$. 
\end{enumerate}
\vspace{.5cm}
\hrule
\vspace{.5cm}

\begin{remark}
	The relabeling algorithm runs in a polynomial time. In fact, its complexity is $O(mn)$. Since the relabeling algorithm works by identifying the set of nonadjacent vertices (adjacent vertices) for each vertex of $V_1$, so the only possible nonadjacent vertices would be the vertices from $V_2$.
\end{remark}

\begin{example}
	We demonstrate the relabeling algorithm on the bipartite graph $G$ given in Figure  \ref{figG}. Although the selection of vertices can be arbitrary, here we implemented the algorithm based on their degrees. Accordingly, the vertices of $V_1$, viz., $a_1,a_2,a_3,a_4$ are relabeled, respectively, as $c_4,c_2,c_3,c_1$. Further,  the vertices of $V_2$, viz., $b_1,b_2,b_3,b_4,b_5$ are relabeled as $c_7,c_9,c_8,c_6,c_5$, respectively. The relabeled graph $G'$ is presented in Figure \ref{figG'}. As per the algorithm, the permutations $w_1, w_2, w_3, w_4$ and $w_5$ of vertices of $G'$ are produced below. Note that, as $N(c_3) = N(c_4) = V_2$, $w_3 = w_4 = \varepsilon$ and further $w_5$ will be nonempty.
	
	\begin{itemize}
		\item[] 	$w_1 = c_4 c_3 c_2 c_5 c_6 c_7 c_1 c_8 c_9$
		\item[] 	$w_2 = c_1 c_3 c_4 c_6 c_5 c_2 c_9 c_8 c_7$
		\item[] 	$w_5 = c_1 c_2 c_3 c_4 c_9 c_8 c_7 c_6 c_5$ 
		\item[]		$w_3 = w_4 = \varepsilon$
	\end{itemize}
	Concatenating all these permutations we get 
	\begin{center}
		$w' = w_1 w_2 w_3 w_4 w_5=c_4 c_3 c_2 c_5 c_6 c_7 c_1 c_8 c_9 c_1 c_3 c_4 c_6 c_5 c_2 c_9 c_8 c_7 c_1 c_2 c_3 c_4 c_9 c_8 c_7 c_6 c_5$
	\end{center} 
	which represents $G'$. Relabeling the vertices in $w'$ with their original labels we get 
	\begin{center}
		$w = a_1 a_3 a_2 b_5 b_4 b_1 a_4 b_3 b_2 a_4 a_3 a_1 b_4 b_5 a_2 b_2 b_3 b_1 a_4 a_2 a_3 a_1 b_2 b_3 b_1 b_4 b_5 $
	\end{center} 
	which represents $G$.

\begin{figure}[!htb]
	\centering
	\begin{minipage}{.5\textwidth}
		\centering
		\[\begin{tikzpicture}
			
			\vertex (a_1) at (1,0) [label=left:$a_1$] {};  
			\vertex (a_2) at (1,-1) [label=left:$a_2$] {};
			\vertex (a_3) at (1,-2) [label=left:$a_3$] {};
			\vertex (a_4) at (1,-3) [label=left:$a_4$] {};
			\vertex (b_1) at (4,0) [label=right:$b_1$] {};  
			\vertex (b_2) at (4,-1) [label=right:$b_2$] {};
			\vertex (b_3) at (4,-2) [label=right:$b_3$] {};
			\vertex (b_4) at (4,-3) [label=right:$b_4$] {};
			\vertex (b_5) at (4,-4) [label=right:$b_5$] {};  
			
			\path
			
			(a_1) edge (b_1)
			(a_1) edge (b_2)
			(a_1) edge (b_3)
			(a_1) edge (b_4)
			(a_1) edge (b_5)
			(a_2) edge (b_1)
			(a_2) edge (b_2)
			(a_2) edge (b_3)
			(a_3) edge (b_1)
			(a_3) edge (b_2)
			(a_3) edge (b_3)
			(a_3) edge (b_4)
			(a_3) edge (b_5)
			(a_4) edge (b_2)
			(a_4) edge (b_3)
			
			;  
			
		\end{tikzpicture}\]
		\caption{A bipartite graph $G$}
		\label{figG}
	\end{minipage}%
	\begin{minipage}{0.5\textwidth}
		\centering
		\[\begin{tikzpicture}
			
			\vertex (c_4) at (1,0) [label=left:$c_4$] {};  
			\vertex (c_2) at (1,-1) [label=left:$c_2$] {};
			\vertex (c_3) at (1,-2) [label=left:$c_3$] {};
			\vertex (c_1) at (1,-3) [label=left:$c_1$] {};
			\vertex (c_7) at (4,0) [label=right:$c_7$] {};  
			\vertex (c_9) at (4,-1) [label=right:$c_9$] {};
			\vertex (c_8) at (4,-2) [label=right:$c_8$] {};
			\vertex (c_6) at (4,-3) [label=right:$c_6$] {};
			\vertex (c_5) at (4,-4) [label=right:$c_5$] {};  
			
			\path
			
			(c_4) edge (c_7)
			(c_4) edge (c_9)
			(c_4) edge (c_8)
			(c_4) edge (c_6)
			(c_4) edge (c_5)
			(c_2) edge (c_7)
			(c_2) edge (c_9)
			(c_2) edge (c_8)
			(c_3) edge (c_7)
			(c_3) edge (c_9)
			(c_3) edge (c_8)
			(c_3) edge (c_6)
			(c_3) edge (c_5)
			(c_1) edge (c_9)
			(c_1) edge (c_8)
			
			;  
			
		\end{tikzpicture}\]
		\caption{Relabeled bipartite graph $G'$}
		\label{figG'}
	\end{minipage}
\end{figure}

\end{example}

\subsection{Correctness of the algorithm}

In the following theorem we establish the correctness of the relabeling algorithm. 
\begin{theorem}\label{main}
	The word $w$ generated by the relabeling algorithm permutationally represents the bipartite graph $G$.
\end{theorem}	

\begin{proof}
If $G$ is a complete bipartite graph, then the algorithm produces the word $w$ given in Step 3. It is a routine verification to observe that the word $w$ represents $G$. 
If $G$ is not a complete bipartite graph, let $G'$ be the graph obtained after relabeling in Step 6. We show that the word $w' = w_1 \cdots w_m w_{m + 1}$ represents $G'$. Accordingly, it is straightforward that $w$ represents $G$.

We now prove that $\overline{ab}$ is an edge in $G'$ if and only if $a$ and $b$ alternate in the word $w'$. 

Suppose $\overline{ab}$ is an edge in $G'$. Without loss of generality assume $a \in V_1$ and $b \in V_2$. 	If $a = c_1$ then clearly $ab \le w_i$ (when $w_i \ne \varepsilon$) for all $i \ge 2$. Further, since $b \in N(a)$ we have $ab \le w_1$ too. Hence, $a$ and $b$ alternate in $w'$. 
If $a = c_p$ for $p \ge 2$, then $ab \le w_1$ and also $ab \le w_{m + 1}$ (when $w_{m+1} \ne \varepsilon$). Further $ab \le a\text{Dec}(N(a)) \le w_p$ (when $w_p \ne \varepsilon$). Hence, $ab \le w_i$ for all nonempty $w_i$ so $a$ and $b$ alternate in $w'$.

Conversely, suppose $\overline{ab}$ is not an edge in $G'$. We deal with this part in three cases, as the case $b \in V_1$ and $a \in V_2$ is symmetric to Case 2. In each case we identify two permutations produced by the algorithm: that one with the subword $ab$ and the other with the subword $ba$. 
	\begin{itemize}
		\setlength\itemsep{1em}
		\item Case 1: $a, b \in V_1$. Let $a = c_p$ and $b = c_q$ with $p > q$. Note $p \ne 1$. Clearly $ab = c_pc_q \le w_1$. If $N(c_p) \ne V_2$, then $w_{p} \ne \varepsilon$ and $ba = c_qc_p \le w_{p}$. Otherwise, $N(c_p) = V_2$ so that $w_{m+1} \ne \varepsilon$ and $ba = c_qc_p \le w_{m+1}$.
		
		\item Case 2: $a \in V_1$ and $b \in V_2$. 
		\begin{itemize}
			\item[-] Subcase 2.1: $a = c_1$. Since $b \notin N(a)$, $ba \le w_1$. Since there are no isolated vertices, $N(b) \ne \varnothing$ and hence there exists $c_q \in V_1$ such that $b \in N(c_q)$. Depending on $N(c_q) = V_2$ or not, $w_{m+1} \ne \varepsilon$ or $w_q \ne \varepsilon$, respectively. Accordingly, $ab \le w_{m+1}$ or $ab \le w_q$.
			
			\item[-] Subcase 2.2: $a = c_p$ for some $p \ge 2$. Clearly, $ab \le w_1$. Since $b \notin N(a)$, we have $N(c_p) \ne V_2$ so that $w_p \ne \varepsilon$. Then $ba = bc_p \le \text{Dec}(V_2 \setminus N(c_p))c_p \le w_p$.
		\end{itemize}
		
		\item Case 3: $a, b \in V_2$. Let $a = c_p$ and $b = c_q$ with $p > q$. Clearly, $ba \le w_1$. As there are no isolated vertices, $N(b) \ne \varnothing$. Hence, there exists $c \in V_1$ such that $b \in N(c)$.
			\begin{itemize}
				\setlength\itemsep{.5em}
				\item[-] If $N(c) = V_2$, then again $w_{m+1} \ne \varepsilon$ and $ab \le w_{m+1}$. 
				\item[-] If $N(c) \ne V_2$, let $c = c_r$ and note that $w_r \ne \varepsilon$. If  $a \in N(c)$, then $ab = c_pc_q \le \text{Dec}(N(c_r)) \le w_r$. Otherwise, $a \in V_2 \setminus N(c)$ so that $ab \le \text{Dec}(V_2 \setminus N(c_r))c_r\text{Dec}(N(c_r)) \le w_r$.
			\end{itemize}
	\end{itemize}
Hence,  $a$ and $b$ do not alternate in $w'$. This completes the proof.

\end{proof}	

\subsection{Upper bounds}

We now obtain an upper bound for the permutation representation number of arbitrary bipartite graphs and also an improved upper bound in case of bipartite graphs with a special property.

First note that, for each vertex in $V_1$, the relabeling algorithm produces at most one permutation of the vertices of $G$. Hence, the algorithm produces a word consisting of at most $m$ permutations of vertices of $G$. Accordingly, we obtain an upper bound for the representation number of bipartite graphs in the following corollary of Theorem \ref{main}.  

\begin{corollary}\label{ub}
	For a bipartite graph $G$, the permutation representation number $\mathcal{R}^p(G) \leq m$. Consequently, $\mathcal{R}(G) \leq m$.
\end{corollary}

Since $\mathcal{R}^p(H_{n,n}) = n$ (for $n \ge 2$), the upper bound obtained in Corollary \ref{ub} is tight for the permutation representation number of the class of bipartite graphs.

\begin{corollary}
	The dimension of a poset corresponding to a bipartite graph is bounded above by $m$.
\end{corollary}

Considering the possibility of occurrence of the empty word $\varepsilon$ in Step 6(b) of the relabeling algorithm, the upper bound given in Corollary \ref{ub} can be improved for a special type of bipartite graphs as per the following corollary.

\begin{corollary}
	Let $\{a_{i_1}, \ldots, a_{i_k}\} \subseteq V_1$ be the set of vertices each of which is adjacent to all vertices of $V_2$ in a bipartite graph $G$, then $\mathcal{R}^p(G) \leq m - k + 1$.
\end{corollary}

\section{Conclusion}

In this paper, we obtained a word $w$ that is a concatenation of permutations of the vertices of a given bipartite graph $G$ through a polynomial time algorithm such that $w$ represent $G$. Accordingly, we obtain a tight upper bound for the permutation representation number of bipartite graphs and hence an upper bound for the representation number of bipartite graphs.  We believe that the relabeling algorithm introduced in this paper has a scope to extend to a larger class of comparability graphs.   

\section*{Acknowledgment} 
We are thankful to the referees for their comments which improved the presentation of the paper. The first author is thankful to the Council of Scientific and Industrial Research (CSIR), Government of India, for awarding the research fellowship for pursuing Ph.D. at IIT Guwahati.

\end{document}